\documentclass[numbib]{imaiai}

\usepackage{caption}
\usepackage{amsmath}
\usepackage{amsthm}
\usepackage{graphicx}
\usepackage{algorithmic}
\usepackage{subfig}
\usepackage{url}
\usepackage{placeins}

\newcommand{\zall}{\boldsymbol{z}_{1:k}}
\newcommand{\zalmost}{\boldsymbol{z}_{1:k-1}}

\newcommand{\dd}{\mathrm{d}}
\newcommand{\bbpp}{\mathbf{p}}

\newcommand{\nrm}  [1] {\Vert #1 \Vert}
\newcommand{\opdiv}{\operatorname{div}}

\begin{document}

\title{Numerical approximation of the Frobenius-Perron operator using the finite volume method}
\shorttitle{Numerical approximation of the Frobenius-Perron operator} 
\shortauthorlist{R.A. Norton, C. Fox and M.E. Morrison} 
\author{{
\sc Richard A. Norton}$^*$,\\[2pt]
Department of Mathematics and Statistics, University of Otago, \\
730 Cumberland Street, Dunedin, New Zealand, \\
$^*${\email{Corresponding author: richard.norton@otago.ac.nz}}\\[2pt]
{\sc and}\\[6pt]
{\sc Colin Fox and Malcolm E. Morrison } \\[2pt]
Department of Physics, University of Otago, \\
730 Cumberland Street, Dunedin, New Zealand.
}

\maketitle

\begin{abstract}
{
We develop a finite-dimensional approximation of the Frobenius-Perron operator using the finite volume method applied to the continuity equation for the evolution of probability. A Courant-Friedrichs-Lewy condition ensures that the approximation satisfies the Markov property, while existing convergence theory for the finite volume method guarantees convergence of the discrete operator to the continuous operator as mesh size tends to zero. Properties of the approximation are demonstrated in a computed example of sequential inference for the state of a low-dimensional mechanical system when observations give rise to multi-modal distributions.
}
{Frobenius-Perron operator; finite volume method; Markov operator; sequential inference.}
\end{abstract}

\section{Introduction}

Consider a dynamical system with state $x\in\mathbb{R}^d$ that evolves according to the ordinary differential equation
\begin{equation}
\label{eq:dynam}
	\frac{\dd x}{\dd t} = v(x)
\end{equation}
where $v(x)$ is a known velocity field.  We have written~\eqref{eq:dynam} in standard autonomous form; if $v = v(x,t)$ then writing $\tilde{x} = (x,t) \in \mathbb{R}^{d+1}$ and $\tilde{v}(\tilde{x}) = (v(x,t),1)$ gives the autonomous form $\tilde{x}_t = \tilde{v}(\tilde{x})$.

Given the velocity field $v(\cdot)$ and an initial state $x_0$, \eqref{eq:dynam} can be solved to determine the future state of the system.  However, if the initial state is uncertain, and is distributed according to a probability density function (pdf) $p_0(x)$, then the future state of the system at time $t>0$ is distributed according to some pdf $p(x|t)$.  The operator that maps the initial pdf $p_0(\cdot)$ to the future pdf $p(\cdot|t)$ is called a \emph{Frobenius-Perron operator}. 

Our interest is in performing sequential inference on the state of the dynamical system from incomplete and uncertain observations of the system at times $0<t_1, t_2, \ldots$, and when the initial state is uncertain. Such problems have been well studied (see, e.g., \cite{Cappe:2005:IHM:1088883}) with well-known solutions for special cases being the Kalman filter and its extensions. In that setting, the Frobenius-Perron operator gives the evolution of pdfs between observations, that is it defines the transition kernel from one time to another. In the language of filtering, the Frobenius-Perron operator defines the \emph{prediction update}, while pdfs between observations enables state estimation or \emph{smoothing}~\cite{Cappe:2005:IHM:1088883}.  This paper is concerned with deriving a numerical approximation to the  Frobenius-Perron operator and understanding the mathematical properties of the resulting discrete operator.

A common route to treating the continuous time system in~\eqref{eq:dynam} is to discretize the dynamics in time (see, e.g., \cite{Golightly2006}) and apply the discrete-time formalism for sequential inference~\cite{Cappe:2005:IHM:1088883}. We take a different approach, instead forming a discrete approximation directly of the continuous time inference problem, and establish  convergence directly on the space of probability distributions.

\subsection{The Frobenius-Perron operator}

Let $X(\cdot,t):\mathbb{R}^d \rightarrow \mathbb{R}^d$ denote the operator that maps an initial condition to the solution of \eqref{eq:dynam} at time $t \geq 0$.  Let $Y(\cdot,t) = X(\cdot,t)^{-1}$ denote its inverse.  If $X(\cdot,t)$ is a \emph{non-singular} transformation ($|Y(E,t)| = 0$ if $|E|=0$ for all Borel subsets $E \subset \mathbb{R}^d$ where $|\cdot|$ denotes Lebesgue measure), then for each $t \geq 0$, the Frobenius-Perron operator $S(t):L^1(\mathbb{R}^d) \rightarrow L^1(\mathbb{R}^d)$ is defined by
$$
	\int_E S(t) f \, \dd x = \int_{Y(E,t)} f \, \dd x \qquad \forall \mbox{ Borel subsets $E \subset \mathbb{R}^d$},
$$
see e.g. \cite[\S 8]{DingZhou}.

Alternatively, given an initial pdf $p_0$, the pdf $p(\cdot|t)$ at some future time, $t>0$, may be computed  by solving the continuity equation (also called the transport or linear advection equation)
\begin{equation}
\label{eq:coneq}
\begin{cases}
	p_t + \opdiv( v p) = 0 & x \in \mathbb{R}^d, \, t > 0, \\
	p(x|0) = p_0(x) & x \in \mathbb{R}^d.
\end{cases}
\end{equation}
Then, for each $t \geq 0$, the Frobenius-Perron operator $S(t): L^1(\mathbb{R}^d) \rightarrow L^1(\mathbb{R}^d)$ is defined such that for any $f \in L^1(\mathbb{R}^d)$,
\begin{equation}
\label{eq:sdef}
	S(t)f := p(\cdot|t) \qquad \mbox{where $p$ is a solution to \eqref{eq:coneq} with $p_0 = f$.}
\end{equation}
See e.g. \cite[Def. 3.2.3 and \S 7.6]{LasotaMackey} or \cite[\S 11.2]{DingZhou} for derivation of the continuity equation from \eqref{eq:dynam}.

The existence of a Frobenius-Perron operator and (weak) solutions to \eqref{eq:coneq} depends on whether or not $X(\cdot,t)$ is non-singular, which depends on the regularity of $v$.    If $v$ has continuous first order derivatives and solutions to \eqref{eq:dynam} exist for all initial points $x_0 \in \mathbb{R}^d$ and all $t \geq 0$ then the Frobenius-Perron operator is well-defined, satisfies the \emph{Markov property}, and $\{S(t): t \geq 0 \}$ defines a continuous semigroup of Frobenius-Perron operators, see \cite[\S 11.2]{DingZhou} or \cite[\S 7.4]{LasotaMackey}.  In such cases the solution to the continuity equation satisfies the ODE $\frac{\dd p}{\dd t} = -(\opdiv v) p$ , along characteristic paths defined by $X(\cdot,t)$.  

\begin{definition} (\cite[Defn. 3.1.1.]{LasotaMackey})
A linear operator $S: L^1(\mathbb{R}^d) \rightarrow L^1(\mathbb{R}^d)$ is a \emph{Markov operator} (or satisfies the \emph{Markov property}) if for any $f \in L^1(\mathbb{R}^d)$ such that $f \geq 0$,
$$
	Sf \geq 0 \qquad \mbox{and} \qquad
	\nrm{Sf}_{L^1(\mathbb{R}^d)} = \nrm{f}_{L^1(\mathbb{R}^d)}.
$$
\end{definition}
We refer to the lesser property $f \geq 0$ implies $Sf \geq 0$ as \emph{positivity preserving}.

The following theorem is due to Liouville (\cite[\S 11.2]{DingZhou},  \cite[\S 7.8]{LasotaMackey}).
\begin{theorem}
\label{thm:stationary}
Let $\{ S(t): t \geq 0 \}$ be the continuous semigroup of Frobenius-Perron operators associated with \eqref{eq:dynam}.  Then for $f \in L^1(\mathbb{R}^d)$ satisfying $\nrm{f}_{L^1(\mathbb{R}^d)} = 1$ and $f \geq 0$, 
$$
	S(t)f=f \quad \forall t \geq 0 \qquad \mbox{if and only if}\qquad \opdiv (f v) = 0.
$$
\end{theorem}
We call the distribution represented by the pdf $f$ in Theorem~\ref{thm:stationary} a \emph{stationary distribution}.
A corollary to Theorem~\ref{thm:stationary} is that the Lebesgue measure is invariant under $X(\cdot,t)$ if and only if $\opdiv v = 0$ (\cite[\S 11.2]{DingZhou}), see~\eqref{eq:vol} below.

\subsection{Hamiltonian systems}

Stationary points of $\{ S(t): t \geq 0 \}$ may not be unique. For example, a Hamiltonian system with Hamiltonian $H: \mathbb{R}^2 \rightarrow \mathbb{R}$ and $v = (\frac{\partial H}{\partial x_2},- \frac{\partial H}{\partial x_1})$, satisfies both $\opdiv(v) = 0$ and $\opdiv(g(H)) = 0$ for differentiable functions $g:\mathbb{R} \rightarrow \mathbb{R}$.   In particular, normalizing $\exp(-H)$ usually yields a stationary distribution.

To prove that $S(t)$ is asymptotically stable, i.e. that $S(t)f$ will converge to a unique stationary distribution as $t \rightarrow \infty$ for any $f$, we can appeal to theorems such as \cite[Thm. 7.4.1]{LasotaMackey}, or \cite[Thm. 13.3.1]{MeynTweedie}.  Loosely speaking, we can expect the conditions for these theorems to be met when the velocity field is such that the system can move between any two states in phase space, and will return to a small volume in phase space in finite time.   

Hamiltonian systems are generally not asymptotically stable because the system state evolves along contours of the Hamiltonian, so the  pdf over future states $S(t)f$ is inextricably linked to the initial $f$.

In the special case when $\opdiv(v) = 0$, e.g. for a Hamiltonian system, then \cite[Thm. 1, Cor. 1.1, Cor. 1.2]{MerletVovelle2007} provides the following precise statements about the properties of solutions to the continuity equation.  

\begin{theorem}
\label{thm:reg1}
If $v \in W^{1,\infty}(\mathbb{R}^d; \mathbb{R}^d)$, $\opdiv (v) = 0$, and $p_0 \in L^1_{loc}(\mathbb{R}^d)$:
\begin{enumerate}
\item Equation \eqref{eq:coneq} has a unique weak solution in the sense that $p \in L^1_{loc}(\mathbb{R}^d \times [0,\infty))$ satisfies
\begin{equation}
\label{eq:weak}
	\int_0^\infty \int_{\mathbb{R}^d} p (\phi_t + v \cdot \nabla \phi) \, \dd x \dd t + \int_{\mathbb{R}^d} p_0 \phi(x,0) \, \dd x = 0 \quad \forall \phi \in C^\infty_0(\mathbb{R}^d \times [0,\infty)).
\end{equation}
\item The weak solution $p$ is constant along characteristic paths defined by $X \in C^1(\mathbb{R}^d \times [0,\infty); \mathbb{R}^d)$, 
\begin{equation}
\label{eq:const}
	p(x|t) = p_0(X(x,t)) \qquad \forall (x,t) \in \mathbb{R}^d \times [0,\infty).
\end{equation}
Volume is also preserved along these characteristic paths in the sense that 
\begin{equation}
\label{eq:vol}
	|X(E,t)| = |E| \qquad \mbox{for every $t \geq 0$ and every Borel subset $E \subset \mathbb{R}^d$}.
\end{equation}
\item For any $0 \leq t \leq T$ the function $X(\cdot,t)$ defining characteristic paths and its inverse $Y(\cdot,t) = X(\cdot,t)^{-1}$ satisfy $X-Id_{\mathbb{R}^d}, Y - Id_{\mathbb{R}^d} \in W^{1,\infty}(\mathbb{R}^d \times [0,T])$ and there exists a constant $C_0 \geq 1$, depending only on $T$, $\nrm{v}_{W^{1,\infty}}$ and $d$, such that 
\begin{equation}
\label{eq:bounds}
\begin{aligned}
	|X(x,t) - x| &\leq C_0 t, & |\nabla X(x,t)| &\leq C_0, & \left| \frac{\partial X}{\partial t}(x,t) \right| & \leq C_0, \\
	|Y(x,t) - x| &\leq C_0 t, & |\nabla Y(x,t)| &\leq C_0, & \left| \frac{\partial Y}{\partial t}(x,t) \right| & \leq C_0.
\end{aligned}
\end{equation}
\item For any measurable function $f : \mathbb{R} \rightarrow \mathbb{R}$ and $p_0 \in L^1_{loc}(\mathbb{R}^d)$ such that $f \circ p_0 \in L^1(\mathbb{R}^d)$,
$$
	\int_{\mathbb{R}^d} f(p(x|t)) \, \dd x = \int_{\mathbb{R}^d} f(p_0(x)) \, \dd x \qquad \forall t \geq 0.
$$
\end{enumerate}
In addition, if $p_0 \in BV(\mathbb{R}^d)$:
\begin{enumerate}
\setcounter{enumi}{4}
\item For any $0 \leq t \leq T$, there exists a constant $C_1$ depending only on $T$, $\nrm{v}_{W^{1,\infty}}$ and $d$, such that 
\begin{align}
	\nrm{ p(\cdot|s) - p(\cdot|t)}_{L^1} &\leq C_1 \nrm{p_0}_{TV} |s-t|, & 0 \leq s,t \leq T, \label{eq:lip}\\
	\nrm{ p(\cdot|t) }_{TV} &\leq C_1 \nrm{p_0}_{TV}, & 0 \leq t \leq T. \label{eq:tv}
\end{align}
\end{enumerate}
\end{theorem}

Immediate consequences of Theorem~\ref{thm:reg1} that are relevant for this article are given in the following Corollary.
\begin{corollary}
\label{cor:0}
If $v \in W^{1,\infty}(\mathbb{R}^d; \mathbb{R}^d)$, $\opdiv (v) = 0$, and $p_0 \in L^1(\mathbb{R}^d)$:
\begin{enumerate}
\item The weak solution $p$ to \eqref{eq:coneq} satisfies $p \in C([0,\infty);L^1(\mathbb{R}^d))$.
\item If $p_0 \geq 0$, then the weak solution $p$ satisfies $p(x|t) \geq 0$ for all $x \in \mathbb{R}^d$ and $t \geq 0$.
\item The weak solution $p$ satisfies $\nrm{p(\cdot|t)}_{L^1(\mathbb{R}^d)} = \nrm{p_0}_{L^1(\mathbb{R}^d)}$ for all $t \geq 0$.
\end{enumerate}
\end{corollary}

\begin{remark}
\label{rem:1}
It follows from Corollary \ref{cor:0} is that the operator $S(t)$, defined by \eqref{eq:sdef}, is a Markov operator for any $t \geq 0$, and that $\{ S(t): t \geq 0 \}$ is a continuous semigroup on $L^1(\mathbb{R}^d)$.  
\end{remark}

\subsection{Numerical approximation of the Frobenius-Perron operator}
For computational purposes it is necessary to numerically approximate the Frobeninus-Perron operator.  Naturally, it is desirable that the approximate operator also has the Markov property, since then the image of an initial pdf remains a positive normalized probability distribution function.

Klus \textit{et al.} \cite{KKSarxiv} provide a good summary of methods that have been used to approximate the Frobenius-Perron operator.

The popular, but slowly converging, \emph{Ulam's method} \cite{Ulam,Li1976,DingZhou} is a Galerkin projection method using piecewise constant basis functions.  It preserves the Markov property but requires computing $|E_i \cap Y(E_j,t)|$ for some $E_i,E_j \subset \mathbb{R}^d$ from a partition of the domain.  The volume of $E_i \cap Y(E_j,t)$ can be estimated using Monte Carlo integration, but this adds computational effort and errors.  The effect of the partition choice and Monte Carlo integration have been investigated in \cite{BoseMurray2001,Murray2004,KoltaiPhD}.  

Faster converging higher order orthogonal Galerkin projections do not, in general, preserve the Markov property \cite{DingLi1991}, but some Petrov-Galerkin methods with piecewise linear or quadratic polynomial basis functions and piecewise constant test functions are able to preserve the Markov property \cite[\S 3]{KoltaiPhD}.  The \emph{piecewise linear Markov finite approximation method} \cite{DingLi1991,DingZhou} is an example of such a Petrov-Galerkin method, that preserves the Markov property and achieves faster convergence with respect to mesh size compared to Ulam's method, see \cite{DingZhou,DingLi1991,DingZhou1995} and references therein. 

Other higher order methods have been used to approximate the Frobenius-Perron operator, see references in \cite{KKSarxiv}, but these methods do not always preserve the Markov property.

The finite volume method (FVM) \cite{eymard,leveque2002,Merlet2007,MerletVovelle2007} is a class of methods that has not yet, to the best of our knowledge, been applied to approximating the Frobenius-Perron operator. This is surprising because FVMs are specifically designed to preserve the conserved quantity of a conservation law, such as the continuity equation \eqref{eq:coneq}. In this case, that property implies preserving the total probability, which is one half of the requirements for the Markov property. Further, FVMs achieve conservation over each local region, or `cell', indicating that convergence in a weak sense is a natural property of the method.

We consider a first-order upwind FVM for solving the continuity equation.  This method inherently preserves the integral of the solution, and also preserves positivity of solutions if the time step satisfies a Courant-Friedrichs-Lewy (CFL) condition. The latter result is given in Lemma~\ref{lem:2}. Thus, the discrete solution operator is a Markov operator.

Another advantage of the FVM we consider is that it comes with existing convergence theory, and we show how the existing theory implies convergence of expectations.  There are higher order FVMs than the one we consider here with accompanying convergence theory.  We expect that some of those higher order methods will also be positivity preserving, provided suitable CFL conditions are satisfied. 

Indeed, positivity preserving FVMs already appear in the literature in disguise.  For example, \cite{Delarue2011} requires a CFL condition for their method to satisfy a `maximum priniciple', and \cite[Prop 2.1]{MerletVovelle2007} has a CFL condition for an `order-preserving' method.  CFL conditions are necessary for stability of FVMs \cite[\S 4.4]{leveque2002}, but stability is a weaker condition than positivity preserving.  

Much of the literature on numerical approximation of Frobenius-Perron operators appears to be concerned with approximating stationary densities. The FVM could be used for that by calculating eigenpairs of the finite-dimensional operator.  

However, as mentioned above, our motivation is to approximate the Frobenius-Perron operator for use as the evolution operator in sequential inference problems, particularly when distributions are non-Gaussain and possibly multi-modal. In such cases, the Kalman filter and its extensions are no longer optimal. The example we consider in Section~\ref{sec:example} considers tracking the position and velocity of a pendulum using incomplete observations and an uncertain initial condition.  The resulting sequence of evolving pdfs are multi-modal.  The implementation we present is feasible for modest sized problems when $d\leq 3$, and possibly when $d=4$. Techniques for meshing problems in higher dimensions may enable implementation in higher dimension, though we do not consider that here.

In the next section we define a first-order upwinding FVM, show that it satisfies the Markov property, and then adapt existing covergence theory to show convergence of expectations with respect the pdf from the approximate Frobenius-Perron operator.  Section \ref{sec:sequential} breifly describes sequential inference 
and presents simulations using the FVM to track the motion of a pendulum from incomplete and uncertain observations.

\section{The Finite Volume Method}
\label{sec:method}

The finite volume method is a family of numerical methods for approximating the solution to partial differential equations (PDEs), including conservation laws.  One of their features, that makes them particularly suitable for conservation laws, is that they preserve (up to round-off error) the conserved quantity in the conservation law.  We propose using a FVM to approximate the Frobenius-Perron operator with a finite-dimensional Markov operator.

FVMs essentially discretize the integral form of the continuity equation \begin{equation}
\label{eq:coneq0}
	\frac{\dd}{\dd t} \int_K p \, \dd x + \oint_{\partial K} p v\cdot n \, \dd S = 0 \qquad \forall K \subset \mathbb{R}^d,
\end{equation}
where $\partial K$ is the boundary of $K$ and $n$ is the outward pointing normal.  A FVM only enforces \eqref{eq:coneq0} on the finite volumes $K$ defined by a mesh $\mathcal{T}$.

Define a mesh $\mathcal{T}$ on $\mathbb{R}^d$ as a family of bounded, open, connected, polygonal, disjoint subsets of $\mathbb{R}^d$ such that $\mathbb{R}^d = \cup_{K \in \mathcal{T}} \overline{K}$.  We refer to each $K \in \mathcal{T}$ as a \emph{cell} or \emph{control volume}.  We also assume that the mesh $\mathcal{T}$ satisfies two properties: the common interface between two cells is a subset of a hyperplane of $\mathbb{R}^d$, and the mesh is admissible (see e.g. \cite[Def. 6.1]{eymard}), so that 
$$
	\exists\alpha > 0 : \begin{cases} \alpha h^d \leq |K| \\ 
	|\partial K| \leq \frac{1}{\alpha} h^{d-1} \end{cases} \qquad \forall K \in \mathcal{T}.
$$
where $h = \sup \{ \operatorname{diam}(K): K \in \mathcal{T} \}$, $|K|$ is the $d$-dimensional Lebesgue measure of $K$, and $|\partial K|$ is the $(d-1)$-dimensional Lebesgue measure of $\partial K$.  If cells $K$ and $L$ have a common interface then we say that $L$ is a neighbour of $K$ and we let $E_{KL}$ denote the interface between $K$ and $L$ and $n_{KL}$ denote the outward pointing normal vector from $K$ into $L$.  Let $\mathcal{N}_K$ denote the set of all neighbours of $K$.  

We discretize the time half-line, $\left\{t:t\geq 0\right\}$, using a regular mesh of size $\Delta t > 0$.  

We use the following first-order upwinding scheme for solving \eqref{eq:coneq} so that we can use the existing theoretical analysis of the FVM in \cite{MerletVovelle2007}.  Define
\begin{equation}
\label{eq:fv0}
	p_K^0 = \frac{1}{|K|} \int_K p_0(x) \, \dd x \qquad \forall K \in \mathcal{T},
\end{equation}
then for $k=0,1,2,\dotsc,$ compute $p_K^{k+1}$ from
\begin{equation}
\label{eq:fv1}
	\frac{p_K^{k+1} - p_K^k}{\Delta t} + \frac{1}{|K|} \sum_{L \in \mathcal{N}_K} v_{KL} p_{KL}^k = 0, \qquad \forall K \in \mathcal{T},
\end{equation}
where
$$
	v_{KL} = \int_{E_{KL}} v \cdot n_{KL} \; \dd S
\qquad \mbox{and} \qquad
	p_{KL}^k = \begin{cases}
		p_K^k & \mbox{if $v_{KL} \geq 0$} \\
		p_L^k & \mbox{if $v_{KL} < 0$}.
	\end{cases}	
$$
The approximate solution $p_h:\mathbb{R}^d \times [0,\infty) \rightarrow \mathbb{R}$ to \eqref{eq:coneq} is then defined as a piecewise constant function satisfying 
\begin{equation}
\label{eq:fv3}
	p_h(x|t) = p_K^k \qquad \forall (x,t) \in K \times [k \Delta t, (k+1)\Delta t).
\end{equation}

Let us also define a discrete Frobenius-Perron operator $S_h(t): L^1(\mathbb{R}^d) \rightarrow V_h$ for $t \geq 0$, where $V_h$ is the restriction of $L^1(\mathbb{R}^d)$ to piecewise constant functions on the mesh $\mathcal{T}$, 
$$
	V_h := \{ f \in L^1(\mathbb{R}^d) : f|_K = \text{constant} \; \forall K \in \mathcal{T} \}.
$$
Then, for every $f \in L^1(\mathbb{R}^d)$ and $t \geq 0$ define 
$$
	S_h(t) f := p_h(\cdot|t) 
$$
where $p_h$ is defined by \eqref{eq:fv0}-\eqref{eq:fv3} with $p_0 = f$.
Note that the operator $S_h(t)$ is mesh dependent and there is a slight abuse of notation because two distinct meshes with the same $h$ do not define the same operator.  Also note that the family of operators $\{ S_h(k \Delta t): k = 0,1,2,\dotsc \}$ defines a semigroup on $L^1(\mathbb{R}^d)$ because we have chosen a regular discretization in $t$. Having the discrete operators form a semigroup is convenient, though not at all necessary for performing sequential inference. 

We can also construct a row vector $\bbpp^{k}$ from the values $\{ |K| p_K^{k} : K \in \mathcal{T} \}$ and write \eqref{eq:fv1} as 
\begin{equation}
\label{eq:fv2}
	\bbpp^{k+1} = \bbpp^{k} \mathbf{S} 
\end{equation}
where $\mathbf{S} = \mathbf{I} - \Delta t \mathbf{A}$ is a countably infinite matrix derived from the matrix  $\mathbf{A}$ that contains the upwind and velocity information from \eqref{eq:fv1}.  

\subsection{Conditions for $S_h(t)$ to be a Markov Operator}
We now determine the conditions under which $S_h(t)$ is a Markov operator.

FVMs preserve the integral of the solution, because they are designed so that the flux leaving a cell $K$ across $E_{KL}$ exactly matches the flux entering its neighbour $L$ across $E_{KL}$.  This means that the integral of $p_h(x|t)$ will remain constant for all $t>0$.

\begin{lemma}
\label{lem:1}
For any $\Delta t > 0$ and any $t \geq 0$,
$$
	\int_{\mathbb{R}^d} p_h(x|t) \, \dd x = \int_{\mathbb{R}^d} p_0(x) \, \dd x.
$$
\end{lemma}

\begin{proof}
First note that \eqref{eq:fv0} implies $\sum_{K \in \mathcal{T}} |K| p_K^0 = \int_{\mathbb{R}^d} p_0(x) \, \dd x$.

Let $K$ and $L$ be neighbouring cells sharing a common interface $E_{KL}$.  Then their outward pointing normal vectors have opposite sign, $n_{KL} = - n_{LK}$, so $v_{KL} = -v_{LK}$ and hence $v_{KL} p_{KL}^k = - v_{LK} p_{LK}^k$.

If we let $\mathcal{E}$ denote the set of all interfaces in the mesh $\mathcal{T}$ then for any $k =0,1,2,\dotsc$, by \eqref{eq:fv1},
\begin{align*}
	\sum_{K \in \mathcal{T}} |K| p_K^{k+1} 
	&= \sum_{K \in \mathcal{T}} |K| p_K^k + \Delta t \sum_{K \in \mathcal{T}} \sum_{L \in \mathcal{N}_K} v_{KL} p_{KL}^k \\
	&= \sum_{K \in \mathcal{T}} |K| p_K^k + \Delta t \sum_{E_{KL} \in \mathcal{E}} (v_{KL} p_{KL}^k + v_{LK} p_{LK}^k) \\
	&= \sum_{K \in \mathcal{T}} |K| p_K^k.
\end{align*}
The result follows by induction and \eqref{eq:fv3}.
\end{proof}

Note that there is no restriction on the mesh or $h$ for Lemma \ref{lem:1} to hold.

This particular FVM also preserves positivity provided $\Delta t$ is sufficiently small.  This immediately follows from \eqref{eq:fv2} and $\mathbf{S} = \mathbf{I} - \Delta t \mathbf{A}$. We now specify a bound for $\Delta t$ depending on $h$, which provides the CFL condition.  Suppose that for some $\xi \in [0,1)$,
\begin{equation}
\label{cfl}
	\Delta t \sum_{L \in \mathcal{N}_K} (v_{KL})_+ \leq (1-\xi) |K| \qquad \forall K \in \mathcal{T},
\end{equation}
where $(v)_+ = \max \{ 0,v \}$.  
The statement $\Delta t \leq (1-\xi) \alpha^2 \nrm{v}_{L^\infty}^{-1} h$ is more recognizable as a CFL condition of the form $\Delta t/h \leq C$, and, by the properties of an admissible mesh, implies that~\eqref{cfl} is satisfied.  

\begin{lemma}
\label{lem:2}
If $\Delta t$ satisfies CFL condition \eqref{cfl} and $p_0 \geq 0$ then
$$
	p_h(x|t) \geq 0 \qquad \forall x \in \mathbb{R}^d, t>0.
$$
\end{lemma}

\begin{proof}
First note that $p_0 \geq 0$ implies $p_K^0 \geq 0$ for all $K \in \mathcal{T}$.  If $p_K^k \geq 0$ for all $K \in \mathcal{T}$, then using \eqref{eq:fv1} and \eqref{cfl}, for every $K \in \mathcal{T}$,
\begin{align*}
	p_K^{k+1} 
	\geq \left( 1 - \frac{\Delta t}{|K|} \sum_{L \in \mathcal{N}_K} (v_{KL})_+ \right) p_K^k 
	\geq 0.
\end{align*}
The result follows by induction and \eqref{eq:fv3}.
\end{proof}

The following Theorem is a direct consequence of Lemmas \ref{lem:1} and \ref{lem:2}.
\begin{theorem}
\label{thm:1}
For any $t \geq 0$, if $\Delta t$ satisfies the CFL condtion \eqref{cfl}, then $S_h(t)$ is a Markov operator.
\end{theorem}
An equivalent statement to Theorem \ref{thm:1} is the following Corollary.
\begin{corollary}
\label{cor:1}
If $\Delta t$ satisfies the CFL condition \eqref{cfl} then the matrix $\mathbf{S}$ from \eqref{eq:fv2} is a stochastic matrix, i.e. the elements of $\mathbf{S}$ are non-negative and the row sums are all equal to $1$.
\end{corollary}

\begin{proof}
We proceed with a proof by contradiction.  Let $\mathbf{S}_{ij}$ denote the entries of $\mathbf{S}$.  Suppose that $\mathbf{S}_{ij}$ is negative for some $i$ and $j$ and let $\mathbf{e}_i$ be the unit row vector with zeros everywhere except the $i^\mathrm{th}$ position.  Then $(\mathbf{e}_i \mathbf{S})_j = \mathbf{S}_{ij} <0$, a contradiction by Lemma \ref{lem:2}.  Therefore, all of the elements of $\mathbf{S}$ are non-negative.  

Now let $\mathbf{s}_i$ denote the $i^\text{th}$ row of $\mathbf{S}$ and suppose $\sum_j (\mathbf{s}_i)_j \neq 1$.  By Lemma \ref{lem:1}, $1 = \sum_j (\mathbf{e}_i)_j = \sum_j ( \mathbf{e}_i \mathbf{S} )_j = \sum_j (\mathbf{s}_i)_j \neq 1$, a contradiction.  Therefore all row sums of $\mathbf{S}$ are equal to $1$.
\end{proof}

If $\mathbf{S}$ is \emph{irreducible} and \emph{aperiodic} (which is the definition of \emph{primitive} for countably infinite matrices), and \emph{positive recurrent}, then $\mathbf{S}$ has a single maximal eigenvalue equal to $1$ with a right eigenvector $\mathbf{1}$ and a left eigenvector that corresponds to the stationary distribution of $S_h(t)$.  If $\mathbf{S}$ is not \emph{aperiodic}, then the eigenvalue may have multiplicity greater than $1$.  For this result and definitions, see \cite[\S 5 and Thm. 1.1]{Seneta}, in particular \cite[Thm. 5.5]{Seneta}.  

Here, irreducibility requires that for any indices $i$ and $j$, there exists a finite sequence of indices $\{ i, i_1, i_2, \dotsc, i_m, j\}$ such that $\mathbf{S}_{ii_1} \mathbf{S}_{i_1i_2} \dotsm \mathbf{S}_{i_mj} > 0$.  This is an analogous to the state evolving from state $i$ to state $j$ in finite time.  Aperiodicity is satisfied for $\mathbf{S}$ if there exists an index $i$ such that $\mathbf{S}_{ii} >0$, which is satisfied if $v_{KL} >0$ for some edge $E_{KL}$.  Positive recurrence requires that a system starting in state $i$ has finite expected return time.  

When the mesh does not align with the contours of the Hamiltonian for a Hamiltonian system, it is possible that $S_h(t)$ is asymptotically stable as $t \rightarrow \infty$, whereas $S(t)$ is not.  This is because the discrete operator effectively allows the discrete system to move between contours of the Hamiltonian, and if a sequence of cells allows traversing all contours then the discrete system may traverse all state space, whereas the continuous system is restricted to contours so cannot traverse state space.

Convergence of a numerical method for a linear PDE is usually a consequence of consistency and stability established by application of the Lax equivalence theorem, see  e.g. \cite[Sec. 8.3.2]{leveque2002}).  Stability of our scheme is a consequence of Theorem \ref{thm:1}, but it is known that first order upwind FVMs lack consistency (the local truncation error does not converge to $0$ as $\Delta t,h \rightarrow 0$) \cite{Bouche2011}.  For this reason, a convergence proof for our FVM is quite delicate.  Here we state the convergence results in \cite{MerletVovelle2007} for our FVM.  Results in \cite{Delarue2011,Merlet2007} and \cite[Chap. 6]{eymard}, and references therein, show convergence results for various first order upwind FVMs and CFL conditions.  

We restrict ourselves to the case when $\opdiv v = 0$ as required by \cite{MerletVovelle2007}.  The theory for the case when $\opdiv v \neq 0$ is more technical, for example `viscosity solutions' or `weak entropy solutions' must be defined to ensure existence and uniqueness of solutions to the continuity equation, see e.g. \cite{leveque2002} and \cite{eymard}. 

When $\opdiv v = 0$, \eqref{eq:fv1} defines the same FVM as in \cite{MerletVovelle2007} since for each $K \in \mathcal{T}$, 
$$
	\sum_{L \in \mathcal{N}_K} v_{KL} p_{KL}^k = \sum_{L \in \mathcal{N}_K} (v_{KL})_+ p_K^k + (v_{KL})_- p_L^k 
	= \sum_{L \in \mathcal{N}_K} (v_{KL})_- (p_L^k - p_K^k),
$$
where $(v)_- = \min \{0,v\}$.  
Also, if $\opdiv v = 0$, $\sum_{L \in \mathcal{N}_K} (v_{KL})_+ = \sum_{L \in \mathcal{N}_K} |(v_{KL})_-|$ and CFL condition \eqref{cfl} is the same condition as \cite[Eqn. 1.9]{MerletVovelle2007}.

From \cite[Thm. 2]{MerletVovelle2007} we have the following Theorem.
\begin{theorem}
\label{thm:2}
Suppose $\opdiv v = 0$ and let $f \in BV(\mathbb{R}^d)$.  If $\Delta t$ satisfies CFL condition \eqref{cfl} for some $\xi \in (0,1)$, then for any $t \geq 0$,
$$
	\nrm{ S(t)f - S_h(t)f }_{L^1(\mathbb{R}^d)} \leq C \xi^{-1} \nrm{f}_{TV} ( t^{1/2} h^{1/2} + \xi^{1/2} t h ).
$$
\end{theorem}
Note $\xi = 0$ is not allowed in the CFL condition for Theorem \ref{thm:2}.

A consequence of convergence of our FVM is the following Theorem.
\begin{theorem}
\label{thm:expconv}
Suppose that $0 \leq t \leq T$ and $p_h \rightarrow p$ in $L^\infty([0,T],L^1(\mathbb{R}^d))$ as $\Delta t,h \rightarrow 0$, and $g \in L^{\infty}_{\mathrm{loc}}(\mathbb{R}^d)$.  If there is a constant $C$ such that
\begin{equation}
\label{eq:33}
	\mathrm{E}_{p(\cdot|t)}[|x|] \leq C, \quad \mathrm{E}_{p(\cdot|t)}[g^2] \leq C, \quad
	\mathrm{E}_{p_h(\cdot|t)}[|x|] \leq C, \quad \mathrm{E}_{p_h(\cdot|t)}[g^2] \leq C,
\end{equation}
independently of $\Delta t$ and $h$, then 
$$
	\mathrm{E}_{p_h(\cdot|t)}[g] \rightarrow \mathrm{E}_{p(\cdot|t)}[g] \qquad \mbox{as $\Delta t,h \rightarrow 0$.}
$$
\end{theorem}

\begin{proof}
For any $R > 0$, let $B_R = \{x \in \mathbb{R}^d: |x| \leq R \}$ and $B_R^c = \mathbb{R}^d \backslash B_R$, and let $\mathbb{I}_S$ be the indicator function for set $S$.  Then, using H\"{o}lder's inequality, the Cauchy-Schwarz inequality, and Markov's inequality,
\begin{align*}
	\left| \mathrm{E}_{p_h(\cdot|t)}[g] - \mathrm{E}_{p(\cdot|t)}[g] \right| 
	&\leq  \int_{B_R} \left | g(x) (p_h(x|t) - p(x|t)) \right| \dd x  + \left| \mathrm{E}_{p_h}[\mathbb{I}_{B_R^c} g] \right| + \left| \mathrm{E}_{p}[\mathbb{I}_{B_R^c} g] \right| \\
	& \leq \nrm{g}_{L^{\infty}(B_R)} \nrm{p_h(\cdot|t)-p(\cdot|t)}_{L^1(\mathbb{R}^d)} + \sqrt{ \mathrm{E}_{p_h}[\mathbb{I}_{B_R^c}] \mathrm{E}_{p_h}[g^2] } + \sqrt{ \mathrm{E}_{p}[\mathbb{I}_{B_R^c}] \mathrm{E}_{p}[g^2]} \\
	& = \nrm{g}_{L^{\infty}(B_R)} \nrm{p_h(\cdot|t)-p(\cdot|t)}_{L^1(\mathbb{R}^d)} + \sqrt{ \mathrm{P}_{p_h}[|x| > R] \mathrm{E}_{p_h}[g^2]} + \sqrt{ \mathrm{P}_{p}[|x|>R] \mathrm{E}_{p}[g^2]} \\
	& \leq \nrm{g}_{L^{\infty}(B_R)} \nrm{p_h(\cdot|t)-p(\cdot|t)}_{L^1(\mathbb{R}^d)} + 2 C R^{-1/2}.
\end{align*}
Then, for any $\epsilon > 0$, first choose $R$ such that $2 C R^{-1/2} < \epsilon /2$, and then choose $\Delta t$ and $h$ such that $\nrm{g}_{L^{\infty}(B_R)} \nrm{p_h(\cdot|t)-p(\cdot|t)}_{L^1(\mathbb{R}^d)} < \epsilon/2$, then $\left| \mathrm{E}_{p_h(\cdot|t)}[g] - \mathrm{E}_{p(\cdot|t)}[g] \right| < \epsilon$, hence result. 
\end{proof}

\begin{remark}
If $p_h$ and $p$ have compact support then \eqref{eq:33} is satisfied for all $g \in L^{\infty}_{loc}$.  If $p_h \rightarrow p$ in $L^\infty([0,T]\times \mathbb{R}^d)$ then the result holds for any $g \in L^1_{\mathrm{loc}}(\mathbb{R}^d)$.
\end{remark}

\begin{remark}
Although Theorem \ref{thm:expconv} does not explicitly require that $\Delta t$ and $h$ satisfy a CFL condition, it is implicit in the assumption that $p_h \rightarrow p$ in $L^\infty([0,T],L^1(\mathbb{R}^2))$ as $\Delta t,h \rightarrow 0$.
\end{remark}

Using Theorem \ref{thm:2} we can specify a convergence result in the case when $p_h$ and $p$ have compact support.

\begin{theorem}
\label{thm:3}
Let $g \in L^{\infty}_{loc}(\mathbb{R}^d)$ and $H,T < \infty$.  If $\opdiv v = 0$, $f \in BV(\mathbb{R}^d)$ has compact support, and there exists $\xi \in (0,1)$ such that $\Delta t$ satisfies the CFL condition \eqref{cfl}, then there exists a constant $C$ independent of $h$ and $t$ such that 
$$
	\left| \mathrm{E}_{S_h(t)f}[g] - \mathrm{E}_{S(t)f}[g] \right| \leq C h^{1/2}  \qquad \forall t \in [0,T], h \in (0,H].
$$ 
\end{theorem}

\begin{proof}
Since $f$ has compact support and $v \in W^{1,\infty}(\mathbb{R}^d; \mathbb{R}^d)$ it follows that $S(t)f$ and $S_h(t)f$ have compact support for all $t \in [0,T]$, so let $B \subset \mathbb{R}^d$ be a bounded set containing the support of $S(t)f$ and $S_h(t)f$ for all $t \in [0,T]$.  Then for all $t \in [0,T]$, by H\"{o}lder's inequality,
$$
	\left| \mathrm{E}_{S_h(t)f}[g] - \mathrm{E}_{S(t)f}[g] \right| \leq \nrm{g}_{L^\infty(B)} \nrm{S_h(t)f - S(t)f}_{L^1(\mathbb{R}^d)}.
$$
The result then follows directly from Theorem \ref{thm:2}.
\end{proof}

In practice, the continuity equation is solved on a bounded domain, so $\mathbb{R}^d$ must be truncated and boundary conditions introduced.  It is relatively easy to show that Lemmas \ref{lem:1} and \ref{lem:2} hold when either periodic boundary conditions or  $v \cdot n = 0$ (Neumann conditions) are specified on a bounded domain.  However, the regularity of $v$ and $p$ could be affected by imposing these boundary conditions, which may have an effect on convergence of the FVM.  Analyses of FVMs on bounded domains is given in~\cite{Vovelle2002}.  See also \cite{eymard}.

Periodic and Neumann boundary conditions prevent probability from entering or leaving the system.  Another possibility, when the support of the initial probability distribution is compact and $v$ is bounded so that the solution always has compact support, is to truncate the domain so that it contains the region of support of the solution and specify homogeneous Dirichl\'{e}t boundary conditions.  If the continuity equation is evolved for only a short time, the region of support remains in the interior of the computational domain, though this strategy may not be practical for long integration times.

When the domain is bounded, then $\mathbf{S}$ in~\eqref{eq:fv2} is a finite dimensional matrix and we can use \cite[Thm. 1.1 or Thm. 1.5]{Seneta} to obtain results about the stationary distributions of $S_h(t)$.  In particular, positive recurrence is irrelevant (because it is guaranteed) for finite matrices $\mathbf{S}$.

The FVM \eqref{eq:fv0}-\eqref{eq:fv3} requires exact integration of $p_0$ over cells to compute $p_K^0$ and integration of $v$ along cell edges to compute $v_{KL}$.  This may not be possible in practice so it may be necessary to replace integration with quadrature for these calculations.  Although convergence analysis of these modified FVMs appears to be an open problem, we have had success (not presented here) in proving that the discrete solution operator is Markov in the case when $v_{KL}$ is computed with the mid-point rule, $v_{KL} = v(x_{KL}) |E_{KL}|$ where $x_{KL}$ is the mid-point of edge $E_{KL}$, using very similar arguments to the proofs of Lemmas \ref{lem:1} and \ref{lem:2}.

Higher order FVMs are also possible, and there is convergence analysis available for these cases, for example see \cite{eymard,leveque2002}.  However, some work is still required to prove that each FVM corresponds to a Markov operator since the existing theory has not been directed to proving that a method is positivity preserving.

\section{An Example of Sequential Inference}
\label{sec:sequential}

Given a sequence of independent noisy observations $\{z_k\}$ made at increasing times $\{t_k\}$, we would like to sequentially infer the state of the observed dynamical system.  Assuming that we have the posterior pdf $\rho(x_{k-1}|\zalmost)$ over the previous state $x_{k-1} = x(t_{k-1})$ (from previous observations $\zalmost$), and we know the likelihood function $\rho(z_k|x_k) = \rho(z_k|x_k,\zalmost)$ for observation $z_k$, then Bayes' theorem gives the following expression for the posterior pdf over $x_k = x(t_k)$ given observations $\zall$,
$$
	\rho(x_k | \zall) = \frac{\rho(z_k|x_k) \rho(x_k | \zalmost) }{ \rho(z_k|\zalmost)}
$$
where 
\begin{align}
\rho(x_k|\zalmost) &= \int \rho(x_k|x_{k-1},\zalmost)\rho(x_{k-1}|\zalmost)\,\dd x_{k-1},\label{eq:top}\\
\rho(z_k|\zalmost) &= \int \rho(z_k|x_k)\rho(x_k|\zalmost)\,\dd x_k.\label{eq:bottom}
\end{align}

Since $\{ x(t) : t \geq 0 \}$ is governed by a dynamical system we can replace \eqref{eq:top} at each step in the sequential inference with the action of the Frobenius-Perron operator
$$
	\rho(x_k|\zalmost) = S(t_k - t_{k-1}) \rho(x_{k-1}|\zalmost).
$$
In the following example we will approximate operation by the Frobenius-Perron operator using the FVM to solve the continuity equation with initial condition given by $\rho(x_{k-1}|\zalmost)$.

\subsection{Tracking a Pendulum}
\label{sec:example}

We now demonstrate the properties of the FVM for use in sequential inference in a simple pendulum example with incomplete and uncertain observations.

For $d=2$, let $x = (x_1,x_2)$ be a vector of angular displacement and angular velocity of a pendulum.  Let $x_1=0$ correspond to the pendulum hanging vertically downwards.  The velocity field for the simple pendulum is
\begin{equation}\label{eq:pendulum}
v(x) = \left(x_2, -\frac{g}{l}\sin(x_1)\right),
\end{equation}
where $g$ is the acceleration due to gravity and $l$ the length of the pendulum.  We use $g = 1$ and $l = 1$.  We also restrict the computational domain to $x \in \Omega = [-\pi,\pi)^2$ with periodic boundary conditions at $x_1 = \pm \pi$ and Neumann boundary conditions at $x_2 = \pm \pi$.  The Frobenius-Perron operator describing the motion of the pendulum is thus well-defined.

To approximate the Frobenius-Perron operator using the FVM we define a uniform lattice of $N \times N$ square cells with side length $h = 2\pi/N$. Then $\sum_{L \in \mathcal{N}_K} (v_{KL}^k)_+ \leq (\pi+1)h$, so the CFL condition \eqref{cfl} is satisfied for some $\xi \in [0,\pi)$ provided that
\begin{equation}
\label{cfl2}
	\Delta t \leq (1-\xi) \frac{h}{\pi + 1}.
\end{equation}
For our simulations we use $\Delta t = (2\pi+1)^{-1} h$ which corresponds to \eqref{cfl2} with $\xi = (2\pi+1)^{-1} \pi$.

With these choices for computational domain, boundary conditions, mesh, and CFL condition, the FVM defines a Markov operator (it's easy to prove similar results to Lemmas \ref{lem:1} and \ref{lem:2}), and, although the conditions for Theorem \ref{thm:2} are not satisfied, we observe in Table \ref{tab:1} that the FVM for this example appears to satisfy the conclusion of Theorem \ref{thm:2}, i.e., a convergence rate of at least order $h^{1/2}$.

\subsubsection{Convergence of the FVM operator}

Table \ref{tab:1} values were computed by evaluating $S(0.25)p_0$ where $p_0$ is a Gaussian pdf with mean $(0.6 \pi,0)$ and covariance $0.64 \mathbf{I}$, discretized and truncated to the computational domain.  The $p_K^0$ values in \eqref{eq:fv0} were approximated with the mid-point rule in each $K$.  No Bayes' steps were used (no observations were made) when computing the values in Table \ref{tab:1}.
\begin{table}
\caption{Table of $L^1(\Omega)$ norms and effective order computed at time $t = \pi/4$. $N = 2\pi/h$ is the number of cells per dimension on the grid.}
\begin{center}
    \begin{tabular}{ | c | c | c | }
    \hline
     & {$L^1$ Norm} & Effective order \\
    $N=1/h$  & $\nrm{p_h - p_{h/2}}_{L^1}$ & $-\log_2\frac{\nrm{ p_{2h} - p_{h} }}{\nrm{p_h - p_{h/2}}}$  \\ \hline
    $50$ & $0.25398$ &  - \\ \hline
    $100$ & $0.19553$ & $0.3855$  \\ \hline
    $200$ & $0.14697$ & $0.4037$ \\ \hline
    $400$ & $0.10799$ & $0.4446$ \\ \hline
    $800$ & $0.062628$ & $0.7860$ \\ \hline
    $1600$& $0.038029$ & $0.7197$ \\ \hline
    $3200$& $0.020514$ & $0.8905$ \\ \hline
    $6400$& - & - \\
    \hline
    \end{tabular}
    \end{center}
\label{tab:1}
\end{table}

The computed values for effective order suggest that the order of our method is approximately $3/4$, but we would need to perform additional simulations with even smaller $h$ to be sure.   Theorem \ref{thm:2} suggests the method has order at least $1/2$, so our observations agree with theory, but they also suggest that Theorem \ref{thm:2} is not sharp for this example.

\subsubsection{Inference from observations}

We simulated the `true' pendulum with initial condition $x=(0.2\pi,0)$, and simulated incomplete and noisy observations of $|x_1|$ at the six times $t_k=k\frac{2 \pi}{7}$, $k=1,2,\ldots,6$ in the interval $t\in[0,2\pi)$, each observation is a draw from $z_k\sim N(|x_1(t_k)|, 0.1^2)$. This situation is analogous to making observations of the force on the pendulum string, which also provides no information about the sign of angular displacement.   We assumed an initial (prior) distribution that is a Gaussian with mean $(0,0)$ and covariance $0.64 \mathbf{I}$, truncated to the computational domain. 

Figure \ref{fig:1} shows the pdf over the angular displacement and velocity of the pendulum resulting from sequential inference at various times.  As time progresses we see that the distribution becomes multimodal due to the ambiguity arising from our deficient observations and symmetric prior.  
\begin{figure}
\begin{center}
\includegraphics[height=8cm]{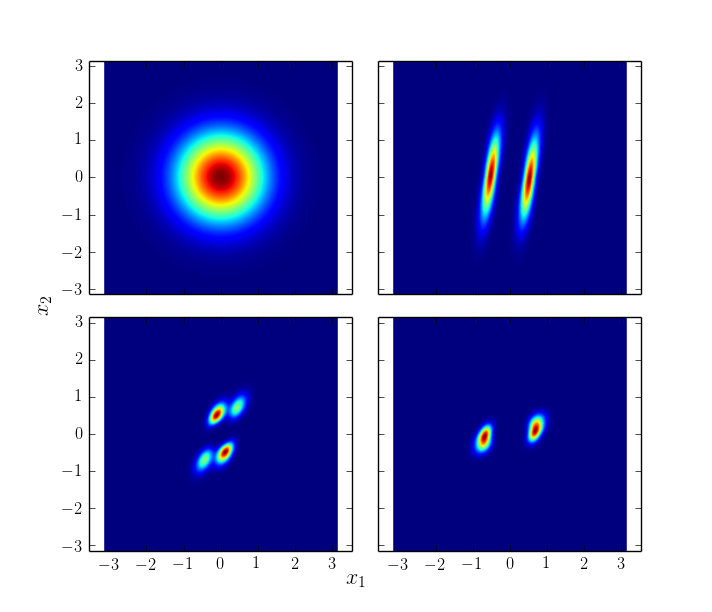} 
\end{center}
\caption{Probability density functions over angular displacement and velocity of the pendulum at times $t=0$ and $t \approx \pi/6$, $\pi/3$ and $\pi$, respectively, top left, top right, bottom left, bottom right.}
\label{fig:1}
\end{figure}

Figure \ref{fig:2} shows the mean and standard deviation of the distributions over $x$ for the pendulum, with the true path also shown. Note that the symmetry of multi-modal distributions means that the mean angular displacement and velocity are always zero, and hence are no use as an estimate of position. In contrast, the multi-modality and location of modes is clear from the pdfs in Figure~\ref{fig:1}. This figure also demonstrates why a sequential inference method that is based on uni-modality, such as the Kalman filter and its extensions, would be a very poor choice for state estimation in this example. It is interesting that the standard deviation, which is the uncertainty in location, eventually provides quite a good estimate for the amplitude of motion of the pendulum.  
\begin{figure}
\begin{center}
\includegraphics[height=7.3cm]{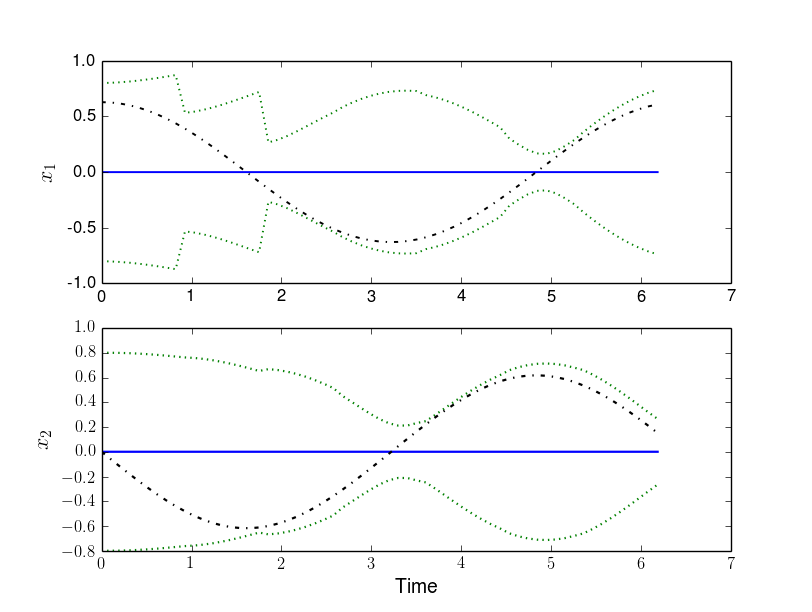} 
\end{center}
\caption{Mean (blue line) and standard deviation (dotted green) about the mean for the sequence of distributions occurring within the sequential inference process, for angular displacement (top) and angular velocity (bottom) of the pendulum. True path is shown with a dot-dashed black line. Steps in the standard deviation are a result of the Bayes step at observations.}
\label{fig:2}
\end{figure}

\section{Conclusion}
\label{sec:conclusion}

We have demonstrated how a FVM can approximate the Frobenius-Perron operator associated with a dynamical system, by approximating the solution to the continuity equation.  

When discretizing the Frobenius-Perron operator it is desirable that the discrete operator also satisfies the Markov property, and we have shown how a FVM achieves this.  

It is perhaps surprising that FVMs have not previously been used to approximate the Frobenius-Perron operator because, in general, they always preserve the conserved quantity of a conservation law.  For the continuity equation this means they always preserve the integral of the solution, one half of the requirements for the Markov property.

The other half of the Markov property requires that the method is positivity preserving.  We showed that the FVM we considered achieves this provided a CFL condition is satisfied.  CFL conditions are not new for FVMs as they are often required for stability.

Another advantage of FVMs is that they often come with pre-existing convergence theory that can be easily adapted to approximating the Frobenius-Perron operator.  

Ulam's method, requires one to compute $|E_i \cap Y(E_j,t)|$ for some $E_i,E_j \subset \mathbb{R}^d$ from a partition of the domain, or approximate these values using Monte Carlo integration.  Similar calculations are required for the piecewise linear Markov finite approximation method and other methods based on Galerkin projection.  The FVM we considered only requires averaging the initial condition over each cell in \eqref{eq:fv0} and evaluating integrals of $v$ along edges between cells to compute $v_{KL}$.  Both of these calculations can be replaced by quadrature to any desired accuracy and their errors included in the analysis.  In this sense, the FVM is easier to implement than some other methods.

With an example we showed how the FVM is well-suited to evolving a pdf that is multi-modal in sequential inference problems.

Particle filters can also be used to perform sequential inference, and typically converge with order ${1/2}$ with respect to the number of particles in the system (see e.g., \cite[\S 2.5]{doucet}).  The FVM presented here only has theoretical order $1/4$ ($3/8$ in practice for the pendulum example) with respect to the number of cells.  However, the FVM presented here could be more efficient than particle filters for some problems.  Also, higher order FVMs and/or adaptive meshes could improve the efficiency of FVMs.

Each inference step with the finite volume approach avoids resampling technicalities of particle filters and we have shown how to adapt existing technology for grid-based PDE solvers to perform grid or mesh-based sequential inference.

We have not proven \emph{Ulam's conjecture} for the FVM approximation of the Frobenius-Perron operator, i.e. that the sequence of stationary distributions for $S_h(t)$ converges to a stationary distribution of $S(t)$ as $h \rightarrow 0$, if $S(t)$ has a stationary distribution.  This is an avenue for further study.

\section*{Funding}

This work was supported by the Royal Society of New Zealnd's Marsden Fund [UOO1015].



\bibliographystyle{imaiai}
\bibliography{nortonbib}

\ifx\undefined\BySame
\newcommand{\BySame}{\leavevmode\rule[.5ex]{3em}{.5pt}\ }
\fi
\ifx\undefined\textsc
\newcommand{\textsc}[1]{{\sc #1}}
\newcommand{\emph}[1]{{\em #1\/}}
\let\tmpsmall\small
\renewcommand{\small}{\tmpsmall\sc}
\fi
\begin{thebibliography}{99}

\bibitem{BoseMurray2001}
\textsc{Bose, C.  {\small \&} Murray, R.}  (2001) The exact rate of
  approximation in {U}lam's method. \emph{Discrete Contin. Dyn. Syst.},
  \textbf{7}(1), 219--235.

\bibitem{Bouche2011}
\textsc{Bouche, D., Ghidaglia, J.-M.  {\small \&} Pascal, F.~P.}  (2011) An
  optimal error estimate for upwind finite volume methods for nonlinear
  hyperbolic conservation laws. \emph{Appl. Numer. Math.}, \textbf{61}(11),
  1114--1131.

\bibitem{Cappe:2005:IHM:1088883}
\textsc{Capp{\'e}, O., Moulines, E.  {\small \&} Ryden, T.}  (2005)
  \emph{{Inference in Hidden {M}arkov Models (Springer Series in Statistics)}}.
  Springer-Verlag New York, Inc.

\bibitem{Delarue2011}
\textsc{Delarue, F.  {\small \&} Lagouti{\`e}re, F.}  (2011) Probabilistic
  analysis of the upwind scheme for transport equations. \emph{Arch. Ration.
  Mech. Anal.}, \textbf{199}(1), 229--268.

\bibitem{DingLi1991}
\textsc{Ding, J.  {\small \&} Li, T.-Y.}  (1991) Markov finite approximation of
  {F}robenius-{P}erron operator. \emph{Nonlinear Anal.}, \textbf{17}(8),
  759--772.

\bibitem{DingZhou}
\textsc{Ding, J.  {\small \&} Zhou, A.}  (2009) \emph{Nonnegative matrices,
  positive operators, and applications}. World Scientific Publishing Co. Pte.
  Ltd., Hackensack, NJ.

\bibitem{DingZhou1995}
\textsc{Ding, J.  {\small \&} Zhou, A.~H.}  (1995) Piecewise linear {M}arkov
  approximations of {F}robenius-{P}erron operators associated with
  multi-dimensional transformations. \emph{Nonlinear Anal.}, \textbf{25}(4),
  399--408.

\bibitem{doucet}
\textsc{Doucet, A., de~Freitas, N.  {\small \&} Gordon, N. } (eds.) (2001)
  \emph{Sequential {M}onte {C}arlo methods in practice}, Statistics for
  Engineering and Information Science. Springer-Verlag, New York.

\bibitem{eymard}
\textsc{Eymard, R., Gallou{\"e}t, T.  {\small \&} Herbin, R.}  (2000) {Finite
  volume methods}. \emph{Handbook of numerical analysis}, \textbf{7},
  713--1018.

\bibitem{Golightly2006}
\textsc{Golightly, A.  {\small \&} Wilkinson, D.~J.}  (2006) Bayesian
  sequential inference for nonlinear multivariate diffusions. \emph{Statistics
  and Computing}, \textbf{16}(4), 323--338.

\bibitem{KKSarxiv}
\textsc{Klus, S., Koltai, P.  {\small \&} Sch\"{u}tte, C.}  (2015) On the
  numerical approximation of the Perron-Frobenius and Koopman operator.
  \emph{ArXiv e-prints}, \textbf{1512.05997}.

\bibitem{KoltaiPhD}
\textsc{Koltai, P.}  (2010) Efficient approximation methods for the global
  long-term behavior of dynamical systems - Theory, algorithms and examples.
  Ph.D. thesis, Technische Universit\"{a}t M\"{u}nchen.

\bibitem{LasotaMackey}
\textsc{Lasota, A.  {\small \&} Mackey, M.~C.}  (1994) \emph{Chaos, fractals,
  and noise}, vol.~97 of \emph{Applied Mathematical Sciences}. Springer-Verlag,
  New York, second edn., Stochastic aspects of dynamics.

\bibitem{leveque2002}
\textsc{LeVeque, R.~J.}  (2002) \emph{{Finite volume methods for hyperbolic
  problems}}, vol.~31. Cambridge university press.

\bibitem{Li1976}
\textsc{Li, T.~Y.}  (1976) Finite approximation for the {F}robenius-{P}erron
  operator. {A} solution to {U}lam's conjecture. \emph{J. Approx. Theory},
  \textbf{17}(2), 177--186.

\bibitem{Merlet2007}
\textsc{Merlet, B.}  (2007) {$L^\infty$}- and {$L^2$}-error estimates for a
  finite volume approximation of linear advection. \emph{SIAM J. Numer. Anal.},
  \textbf{46}(1), 124--150.

\bibitem{MerletVovelle2007}
\textsc{Merlet, B.  {\small \&} Vovelle, J.}  (2007) Error estimate for finite
  volume scheme. \emph{Numer. Math.}, \textbf{106}(1), 129--155.

\bibitem{MeynTweedie}
\textsc{Meyn, S.  {\small \&} Tweedie, R.~L.}  (1993) \emph{Markov chains and
  stochastic stability}. Springer-Verlag, London.

\bibitem{Murray2004}
\textsc{Murray, R.}  (2004) Optimal partition choice for invariant measure
  approximation for one-dimensional maps. \emph{Nonlinearity}, \textbf{17}(5),
  1623--1644.

\bibitem{Seneta}
\textsc{Seneta, E.}  (2006) \emph{Non-negative matrices and {M}arkov chains},
  Springer Series in Statistics. Springer, New York, Revised reprint of the
  second (1981) edition [Springer-Verlag, New York; MR0719544].

\bibitem{Ulam}
\textsc{Ulam, S.~M.}  (1960) \emph{A collection of mathematical problems},
  Interscience Tracts in Pure and Applied Mathematics, no. 8. Interscience
  Publishers, New York-London.

\bibitem{Vovelle2002}
\textsc{Vovelle, J.}  (2002) Convergence of finite volume monotone schemes for
  scalar conservation laws on bounded domains. \emph{Numer. Math.},
  \textbf{90}(3), 563--596.

\end{thebibliography}

\end{document}